\documentclass[pra,groupedaddress,superscriptaddress,twocolumn,showpacs,amsmath,amssymb,a4paper]{revtex4}
\usepackage{geometry}
\usepackage{graphicx}
\usepackage{mathrsfs}
\usepackage{amssymb}
\usepackage{amsmath}
\usepackage{amsthm}
\usepackage{bm}
\usepackage{hyperref}

\newtheorem{proposition}{Proposition}

\newtheorem{corollary}{Corollary}

\theoremstyle{definition}

\newcommand{\bra}[1]{\langle #1|}
\newcommand{\ket}[1]{| #1 \rangle }
\newcommand{\ip}[2]{{\langle #1|}{ #2 \rangle }}
\newcommand{\tr}[1]{{\rm tr}[#1]}
\newcommand{\be}{\begin{eqnarray}}
\newcommand{\ee}{\end{eqnarray}}

\newcommand{\cE}{{\cal E}}
\newcommand{\cG}{{\cal G}}

\newcommand{\cT}{{\cal T}}
\newcommand{\cC}{{\cal C}}
\newcommand{\cS}{{\cal S}}
\newcommand{\cH}{{\cal H}}

\begin{document}

\title{Entanglement sensitivity to signal attenuation and amplification}

\author{Sergey N. Filippov}

\affiliation{Moscow Institute of Physics and Technology,
Institutskii Per. 9, Dolgoprudny, Moscow Region 141700, Russia}

\affiliation{Institute of Physics and Technology, Russian Academy
of Sciences, Nakhimovskii Pr. 34, Moscow 117218, Russia}

\affiliation{Russian Quantum Center, Novaya 100, Skolkovo, Moscow
Region 143025, Russia}

\author{M\'{a}rio Ziman}

\affiliation{Institute of Physics, Slovak Academy of Sciences,
D\'{u}bravsk\'{a} cesta 9, Bratislava 84511, Slovakia}

\affiliation{Faculty of Informatics, Masaryk University,
Botanick\'{a} 68a, Brno 60200, Czech Republic}

\begin{abstract}
We analyze general laws of continuous-variable entanglement
dynamics during the deterministic attenuation and amplification of
the physical signal carrying the entanglement. These processes are
inevitably accompanied by noises, so we find fundamental
limitations on noise intensities that destroy entanglement of
gaussian and non-gaussian input states. The phase-insensitive
amplification $\Phi_1 \otimes \Phi_2 \otimes \ldots \Phi_N$ with
the power gain $\kappa_i \ge 2$ ($\approx 3$~dB, $i=1,\ldots,N$)
is shown to destroy entanglement of any $N$-mode gaussian state
even in the case of quantum limited performance. In contrast, we
demonstrate non-gaussian states with the energy of a few photons
such that their entanglement survives within a wide range of
noises beyond quantum limited performance for any degree of
attenuation or gain. We detect entanglement preservation
properties of the channel $\Phi_1 \otimes \Phi_2$, where each mode
is deterministically attenuated or amplified. Gaussian states of
high energy are shown to be robust to very asymmetric
attenuations, whereas non-gaussian states are at an advantage in
the case of symmetric attenuation and general amplification. If
$\Phi_1 = \Phi_2$, the total noise should not exceed
$\frac{1}{2}\sqrt{\kappa^2+1}$ to guarantee entanglement
preservation.
\end{abstract}

\pacs{03.67.Mn, 03.65.Ud, 03.65.Yz}

\maketitle

Creation, manipulation, and evolution of entangled states are in
the basis of many applications including quantum information
protocols \cite{horodecki-2009} and interferometry
\cite{giovannetti-lloyd-maccone-2011}. The physical implementation
of such applications raises an important problem of noisy
entanglement dynamics and robustness of entangled states
\cite{banaszek-2009}. The problem of continuous-variable
entanglement dynamics in different physical models of
system--environment interactions was considered in the papers
\cite{serafini-2004,manko-2005,adesso-illuminati-2007,an-2007,paz-2008,isar-2009,vasile-2009,galve-2010,argentieri-2011,barbosa-2011,buono-2012,marzolino-2013,rebon-2014,aolita-2014}.
The particular results depend on many aspects, namely, the
structure of composite system comprising (in)distinguishable
particles, the entanglement measure, the noise model, the initial
state, the interaction among particles of the system, and the form
of external driving. Such a variety of scenarios makes the full
characterization of entanglement dynamics hardly possible.
Moreover, the extent to which the evolved entanglement remains
useful depends on the particular quantum application. However, a
general entanglement-assisted application relies on the presence
of non-vanishing entanglement, an exceptional quantum property
regardless of its magnitude. Thus, the fundamental limitation on
the application performance is imposed by those noises that
completely destroy the entanglement of an input state.

In this paper, we analyze the limiting noises that accompany the
physical processes of deterministic signal attenuation and
amplification \cite{caves-1982}. The former one is a standard
model to describe losses in continuous-variable systems
\cite{barbosa-2011,buono-2012}, whereas the latter one is used in
so-called quantum cloning machines \cite{scarani-2005} and other
applications \cite{weedbrook-2012}. The limiting noises for such
operations were found in the one-sided scenario, i.e. for a
quantum channel of the form $\Phi_1\otimes{\rm Id}_2$, which
transforms any input state into a separable one
\cite{holevo-2008}. Such quantum channels $\Phi$ are known as
entanglement breaking ones \cite{holevo-1998,horodecki-2003}.
However, the attenuation or amplification does not have to be
one-sided. Our goal is to find parameters of the general channel
$\Phi_1 \otimes \Phi_2 \otimes \cdots$ that fundamentally restrict
the use of locally attenuated or amplified signals in
entanglement-assisted applications.

This problem was partially explored in the paper
\cite{sabapathy-2011}, which announced the existence of
non-gaussian states that are more robust to the action of
homogeneous two-mode amplification $\Phi\otimes\Phi$ than gaussian
ones (in contrast to the beforehand opposite conjecture
\cite{allegra-2010,adesso-2011,lee-2011,allegra-2011}). Our
results improve those of Ref.~\cite{sabapathy-2011} and provide
evidence that non-gaussian states of little energy can outperform
high-energy gaussian states also in the case of two-mode
attenuation. We extend our results to asymmetric channels and
multiple numbers of modes.

{\it Attenuators and amplifiers} are distinguished examples of
gaussian channels
\cite{holevo-werner-2001,caruso-2006,holevo-2007} that are usually
used to describe the deterministic lossy process and linear
amplification of bosonic quantum states. The bosonic quantum state
is defined by the density operator $\varrho$ or, equivalently, by
the characteristic function $\varphi({\bf z}) = \tr{\varrho W({\bf
z})}$, where $W({\bf z})=\exp[i(q_1 x_1 +p_1 y_1 + \cdots + q_N
x_N + p_N y_N)]$ is the Weyl operator, $N$ is the number of modes,
the operators $q_i$ and $p_j$ satisfy the canonical commutation
relation $[q_i,p_j]=i\delta_{ij}$, and ${\bf z} =
(x_1,y_1,\ldots,x_N,y_N)^{\top}$ corresponds to coordinates in the
real symplectic space $(\mathbb{R}^{2N},\bm{\Delta})$, with
$\bm{\Delta}$ being the symplectic form $\bm{\Delta} =
\bigoplus\limits_{i=1}^{N} \left(
\begin{array}{cc}
  0 & -1 \\
  1 & 0 \\
\end{array}
\right)$.

In terms of the characteristic functions, the gaussian channel
acts as follows:
\begin{equation}
\label{attenuator-amplifier-channel} \varphi_{\rm out}({\bf z}) =
\varphi_{\rm in}({\bf K}{\bf z}) \exp \left( -\frac{1}{2} {\bf
z}^{\top} {\bf M} {\bf z} \right).
\end{equation}

\noindent The one-mode gaussian channels are characterized in
\cite{holevo-2007}. Suppose ${\bf K} = \sqrt{\kappa} \left(
\begin{array}{cc}
  1 & 0 \\
  0 & 1 \\
\end{array}
\right)$ and ${\bf M} = \mu \left(
\begin{array}{cc}
  1 & 0 \\
  0 & 1 \\
\end{array}
\right)$, then the transformation
\eqref{attenuator-amplifier-channel} defines processes of one-mode
attenuation ($0<\kappa<1$), addition of classical noise
($\kappa=1$), and amplification ($\kappa>1$). These processes are
fair physical channels (completely positive maps) if the total
noise $\mu \ge \frac{1}{2}|\kappa-1|$ \cite{holevo-werner-2001}.
The minimal noise $\mu_{\rm QL}=\frac{1}{2}|\kappa-1|$ corresponds
to a so-called quantum limited operation, and the quantity $a =
\mu - \mu_{\rm QL} \ge 0$ is the extra noise. The one-mode
channels $\Phi(\kappa,\mu)$ altogether form a set $\cC$.

The action of channel $\Phi\in\cC$ takes a simple form in the
diagonal sum representation $\Phi[X] = \sum_{ij} A_{ij} X
A_{ij}^{\dag}$, where the explicit form of Kraus operators
$\{A_{ij}\}_{i,j=0,1,\ldots}$ in the Fock basis has been found for
all $\kappa$ and $a$ in the seminal paper \cite{ivan-2011}. To
work easily with the coherent states $\ket{\alpha}$, we derive the
representation $\Phi[\varrho] = \pi^{-2} \iint d^2 \alpha \, d^2
\beta \, \tilde{A}_{\alpha\beta} \varrho
\tilde{A}_{\alpha\beta}^{\dag}$, where
\begin{align}
\label{kraus} &\tilde{A}_{\alpha\beta} = \int \frac{d^2
\gamma}{\pi\sqrt{\tau}} \exp \bigg(
-\frac{|\alpha|^2+|\beta|^2+|\gamma|^2}{2} + \sqrt{1-\eta} \,
\alpha\gamma \nonumber\\
&+ \frac{1}{2\tau} \left| \sqrt{\tau - 1} \, \beta + \sqrt{\eta}
\, \gamma \right|^2  \bigg) \Big| \sqrt{\tfrac{\tau - 1}{\tau}} \,
\beta+\sqrt{\tfrac{\eta}{\tau}} \, \gamma \Big\rangle
\bra{\gamma}, \\
\label{k1-k2} & \qquad\quad \eta = \frac{\kappa}{\tau}, \qquad
\tau = \left\{
\begin{array}{ll}
  1+a, & 0<\kappa<1,\\
  \kappa+a, & \kappa>1.\\
\end{array}
\right.
\end{align}

\noindent The parameter $\eta$ defines the attenuation factor of
the quantum limited attenuation $\Phi_{\rm QL \eta}$ and $\tau$
defines the power gain of the quantum limited amplifier $\Phi_{\rm
QL \tau}$, the concatenation of these channels results in the
channel $\Phi(\kappa,\mu)$ given by
\eqref{attenuator-amplifier-channel}, i.e. $\Phi_{\rm QL \tau}
\circ \Phi_{\rm QL \eta} = \Phi(\kappa,\mu)$ \cite{ivan-2011}.

{\it Entanglement annihilation}. The phenomenon of complete
entanglement degradation is known as entanglement annihilation
\cite{moravcikova-ziman-2010} and was analyzed for discrete
variable systems in the papers
\cite{filippov-rybar-ziman-2012,filippov-ziman-2013,filippov-melnikov-ziman-2013}.
The density operator $\varrho$ acting on $\cH^{\otimes N}$ is
called mode-entangled (separable) if it cannot (can) be
represented as a convex sum $\sum_i p_i \varrho_i^{(1)} \otimes
\cdots \otimes \varrho_i^{(N)}$, where $p_i \ge 0$,
$\varrho_i^{(j)} \ge 0$. The channel $\Upsilon: \cT(\cH^{\otimes
N}) \mapsto \cT(\cH^{\otimes N})$ annihilates entanglement of some
input state $\varrho$, if the output state $\Upsilon[\varrho]$ is
separable. If this property holds true for all $\varrho$ from some
domain $\cS$, then $\Upsilon$ is called entanglement annihilating
on $\cS$.

The channel $\Phi^{\otimes N}$ is of a great interest in quantum
communication: The parts of a composite quantum state (encoded in
time bins or radiation modes) are sent through the same
communication line modeled by the channel $\Phi$
\cite{holevo-giovannetti}. If $\Upsilon=\Phi^{\otimes N}$ is
entanglement annihilating on $\cS$, then $\Phi$ is called
$N$-locally entanglement annihilating ($N\text{-LEA}_{\cS}$). We
will refer to the channel $\Phi$ as $\infty$-LEA if $\Phi \in
N$-LEA for all $N=2,3,\ldots$. If $\Phi$ is entanglement breaking,
then it is a measure and prepare operation
\cite{holevo-1998,horodecki-2003} that definitely disentangles the
part it acts on from all other parts of the multipartite system.
The inclusion diagram follows \cite{moravcikova-ziman-2010}:
\begin{equation}
\label{inclusion-general} \text{EB} \subset
\infty\text{-LEA}_{\cS} \subset \cdots \subset 3\text{-LEA}_{\cS}
\subset 2\text{-LEA}_{\cS}.
\end{equation}

\noindent Further, this relation will be specified for the
channels from class $\cC$.

{\it Gaussian input states} find applications in many quantum
information protocols \cite{weedbrook-2012}. The characteristic
function of a gaussian state $\varrho\in\cG$ reads $\varphi({\bf
z}) = \exp\left(-\frac{1}{2} {\bf z}^{\top} {\bf V} {\bf z} + i
{\bf l}^{\top} {\bf z} \right)$, where ${\bf V} \ge
\frac{i}{2}\bm{\Delta}$ is the covariance matrix and ${\bf l}$ is
the vector of average values of $q_i$, $p_i$, $i=1,\ldots,N$. The
vector ${\bf l}$ is irrelevant for entanglement properties, so we
let ${\bf l}=0$.

\begin{figure}
\includegraphics[width=8.5cm]{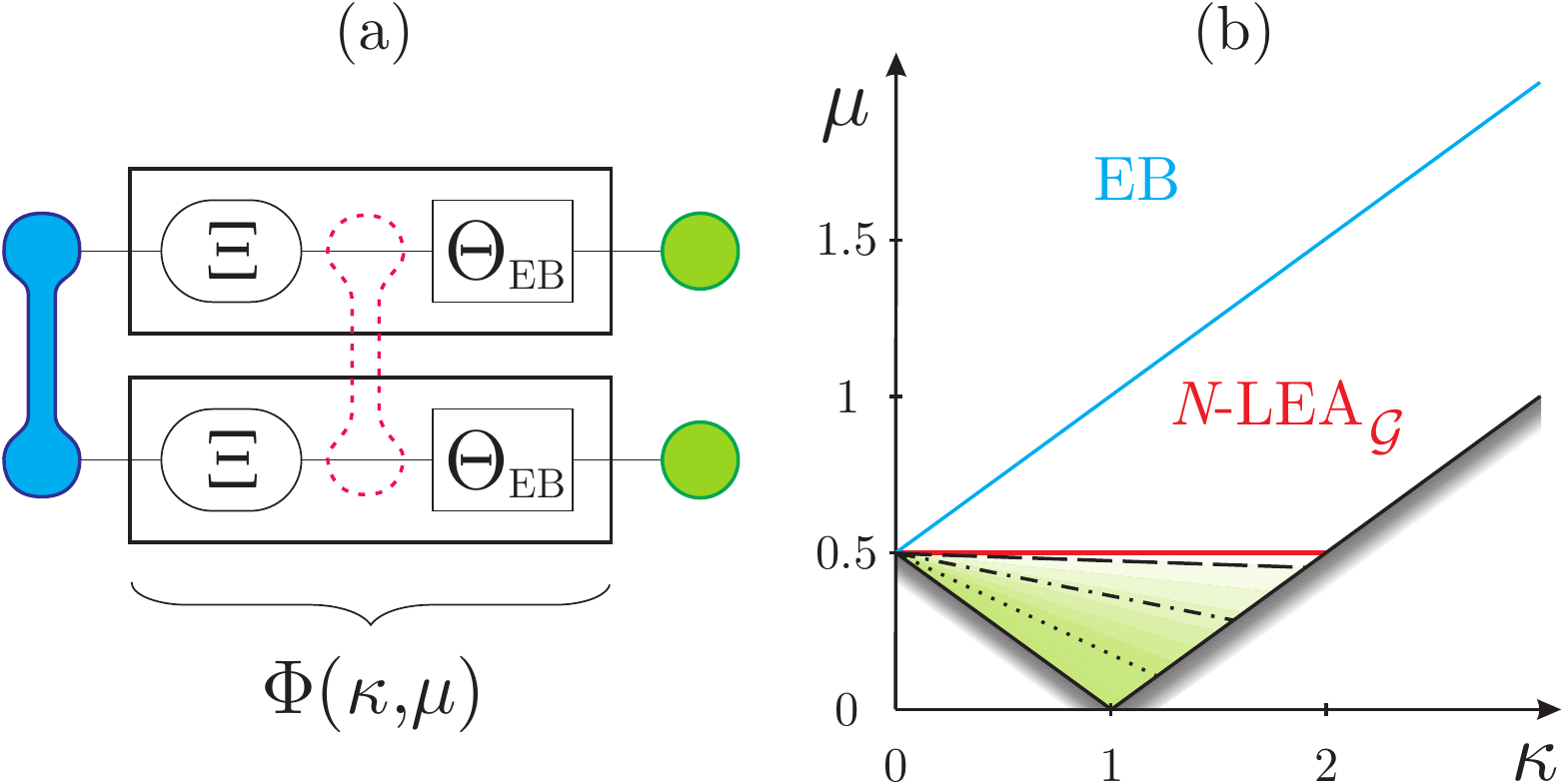}
\caption{\label{figure1} (Color online) (a) Decomposition of the
one-mode channel $\Phi(\kappa,\mu)$ into the scaling map $\Xi$
with $\kappa_{\Xi} \gg 1$ and $\mu_{\Xi}=0$ and the entanglement
breaking map $\Theta_{\rm EB}$. Action of the map $\Xi^{\otimes
N}$ on any gaussian input results in a valid gaussian state (red
dotted line), whose entanglement is then annihilated by the
entanglement breaking maps. (b) Map $\Phi(\kappa,\mu)$ is a valid
channel above gray shading. Channel $\Phi(\kappa,\mu)^{\otimes N}$
is entanglement annihilating for points $(\kappa,\mu)$ above the
horizontal red solid line for every $N=2,3,\ldots$. Entanglement
of the two-mode squeezed state with energy $\cE$ (measured in
photons) is preserved by the channel $\Phi(\kappa,\mu)^{\otimes
2}$ below the lines: dotted if $\cE=0.1$, dash-dotted if $\cE=1$,
and dashed if $\cE=10$.}
\end{figure}

\begin{proposition}
\label{proposition-1} The channel $\Phi(\kappa,\mu) \in \cC$ is
$N\text{-LEA}_{\cG}$ for all $N=2,3,\ldots$ if and only if the
total noise level $\mu \ge \frac{1}{2}$.
\end{proposition}
\begin{proof}
Let us verify when the channel $\Phi(\kappa,\mu)$ can be
represented as a concatenation $\Theta_{\rm EB} \circ \Xi$ of the
scaling map $\Xi$ given by formula
\eqref{attenuator-amplifier-channel} with $\kappa_{\Xi} \gg 1$,
${\bf M} = \bm{0}$, and the entanglement breaking attenuator
$\Theta_{\rm EB}$ with $\kappa_{\Theta} \ll 1$ [Fig.
\ref{figure1}(a)]. The scaling map $\Xi$ is not positive in
general but it transforms any gaussian state into another gaussian
state because the transformed covariance matrix satisfies the
condition ${\bf V}_{\rm out} = \kappa_{\Xi} {\bf V}_{\rm in} \ge
\frac{i}{2}\bm{\Delta}$. The relation $\Phi(\kappa,\mu) =
\Theta_{\rm EB} \circ \Xi$ holds if $\kappa = \kappa_{\Xi}
\kappa_{\Theta}$ and $\mu = \mu_{\Theta} \ge
\frac{1}{2}(1+\kappa_{\Theta})$, the latter inequality being a
necessary and sufficient condition for the entanglement breaking
property of $\Theta$ \cite{holevo-2008}. In the limit
$\kappa_{\Theta} \rightarrow 0$ and $\kappa_{\Xi} \rightarrow
\infty$ with keeping $\kappa_{\Xi}\kappa_{\Theta} = \kappa= {\rm
const}$, we obtain $\mu \ge \frac{1}{2}$. Thus, the channel
$\Phi(\kappa,\mu)$ is a concatenation of the scaling map (positive
on gaussian inputs) and entanglement breaking map if $\mu \ge
\frac{1}{2}$. Those entanglement breaking maps make the output
state separable, which proves sufficiency.

If $\mu<\frac{1}{2}$ then $\Phi(\kappa,\mu)^{\otimes 2}$ preserves
entanglement of the two-mode squeezed vacuum state
$\ket{\psi}=\sqrt{1-\tanh^2 r}\sum_{n=0}^{\infty} (\tanh r)^n
\ket{n}\otimes\ket{n}$ when $r \rightarrow \infty$. This can be
checked, e.g., by Simon's criterion~\cite{simon-2000} applied to
the covariance matrix
\begin{equation}
\label{covariance-squeezed} {\bf V} = \frac{1}{2} \left(
\begin{array}{cccc}
  \cosh 2r & 0 & \sinh 2r & 0 \\
  0 & \cosh 2r & 0 & -\sinh 2r \\
  \sinh 2r & 0 & \cosh 2r & 0 \\
  0 & -\sinh 2r & 0 & \cosh 2r \\
\end{array}
\right).
\end{equation}

\noindent Thus, $\Phi(\kappa,\mu)$ is not $2\text{-LEA}_{\cG}$
and, consequently, not $N\text{-LEA}_{\cG}$. This proves the
necessity.
\end{proof}

We emphasize that Proposition~\ref{proposition-1} is valid for any
$N=2,3,\ldots$, which complements the previously known result for
$N=2$ \cite{sabapathy-2011} and is in agreement with the
attenuation experiment of Ref. \cite{buono-2012}.

$\Phi(\kappa,\mu)$ is entanglement breaking if and only if $a \ge
\min(\kappa,1)$ \cite{holevo-2008}. This means that
$\text{EB}^{\cC} \ne \infty\text{-LEA}_{\cG}^{\cC}$ and the
inclusion diagram \eqref{inclusion-general} takes the following
form for $\cC$-channels and gaussian inputs:
\begin{equation}
\text{EB}^{\cC} \subsetneq \infty\text{-LEA}_{\cG}^{\cC} = \cdots
= 3\text{-LEA}_{\cG}^{\cC} = 2\text{-LEA}_{\cG}^{\cC}.
\end{equation}

Gaussian state entanglement cannot survive the amplification with
$\kappa>2 \approx 3$~dB. Even if the total noise
$\mu<\frac{1}{2}$, the gaussian state should have enough energy to
protect its entanglement from annihilation. The two-mode squeezed
vacuum has energy $\cE = \cosh 2r - 1$ photons and its
entanglement is annihilated by $\Phi(\kappa,\mu)^{\otimes 2}$
unless $\mu < \frac{1}{2}\left[
1-\kappa+\kappa\left(\sqrt{\cE(2+\cE)}-\cE\right) \right]$ [see
Fig.~\ref{figure1}(b)].

The result of Proposition~\ref{proposition-1} can be extended to
the case of nonhomogeneous local channels.

\begin{corollary}
\label{corollary-1} The channel $\Upsilon=\Phi(\kappa_1,\mu_1)
\otimes \Phi(\kappa_2,\mu_2)\otimes \cdots \otimes
\Phi(\kappa_N,\mu_N)$ annihilates entanglement of all gaussian
$N$-mode states if $\mu_i \ge \frac{1}{2}$, $i=1,\ldots,N$.
\end{corollary}
\begin{proof} Similarly to the proof of
Proposition~\ref{proposition-1}, concatenation of the homogeneous
scaling map $\Xi^{\otimes N}$ and a map $\bigotimes_{i=1}^{N}
\Theta_{{\rm EB} i}$ composed of the individual
entanglement-breaking attenuators with $a_{\Theta i} =
\kappa_{\Theta i} = \kappa_i / \kappa_{\Xi}$ leads to the map
$\Phi(\kappa_1,\frac{1}{2}) \otimes
\Phi(\kappa_2,\frac{1}{2})\otimes \cdots \otimes
\Phi(\kappa_N,\frac{1}{2})$ in the limit $\kappa_{\Xi}\rightarrow
\infty$. The map $\Upsilon$ may be readily obtained by adding
classical noise $(\mu_i-\frac{1}{2})$ into $i$th mode,
$i=1,\ldots,N$.
\end{proof}

\begin{corollary}
\label{corollary-2} Suppose the channel
$\Upsilon=\Phi(\kappa_1,\mu_1) \otimes \Phi(\kappa_2,\mu_2)\otimes
\cdots \otimes \Phi(\kappa_N,\mu_N)$ such that
$\min\limits_{i=1,\ldots,N-1} \frac{2\mu_i-1}{\kappa_i} = s \ge
0$. The channel $\Upsilon$ annihilates entanglement of all
gaussian $N$-mode states if $\mu_N \ge \frac{1}{2}(1-s\kappa_N)$.
\end{corollary}
\begin{proof}
If $s \ge 1$, then all the channels $\Phi(\kappa_1,\mu_1)$,
$\Phi(\kappa_2,\mu_2)$, $\ldots$, $\Phi(\kappa_{N-1},\mu_{N-1})$
are entanglement breaking and the statement becomes trivial. If
$0<s<1$, let us represent $\Upsilon$ as a concatenation of the
homogeneous scaling map $\Xi^{\otimes N}$ and a map
$\left(\bigotimes_{i=1}^{N-1} \Theta_{{\rm EB} i}\right)\otimes
\Theta_{{\rm QL} N}$ composed of $N-1$ entanglement breaking
attenuators and a quantum limited attenuation of the $N$th mode
[see Fig.~\ref{figure2}(a)]. In fact, put $\kappa_{\Xi} = s^{-1} >
1$ and $\kappa_{\Theta i} = s \kappa_i$, $i=1,\ldots,N-1$, then
the relation $a_{\Theta i} = \mu_i - \frac{1}{2}(1-s\kappa_i) \ge
\kappa_{\Theta i}$ makes $\Theta_i$ entanglement breaking for
$i=1,\ldots,N-1$, which guarantees separability of the output
state for all $N$-mode gaussian inputs. The application of the
quantum limited attenuator $\Theta_{{\rm QL} N}$ with
$\kappa_{\Theta N} = s \kappa_N$ results in the noise $\mu_N =
\frac{1}{2}(1-\kappa_{\Theta N}) = \frac{1}{2}(1-s\kappa_N)$.
Greater noises in $N$th mode can be realized by adding classical
noise. The case $s=0$ corresponds to the closure of the set of
separable states.
\end{proof}

\begin{figure}
\includegraphics[width=8.5cm]{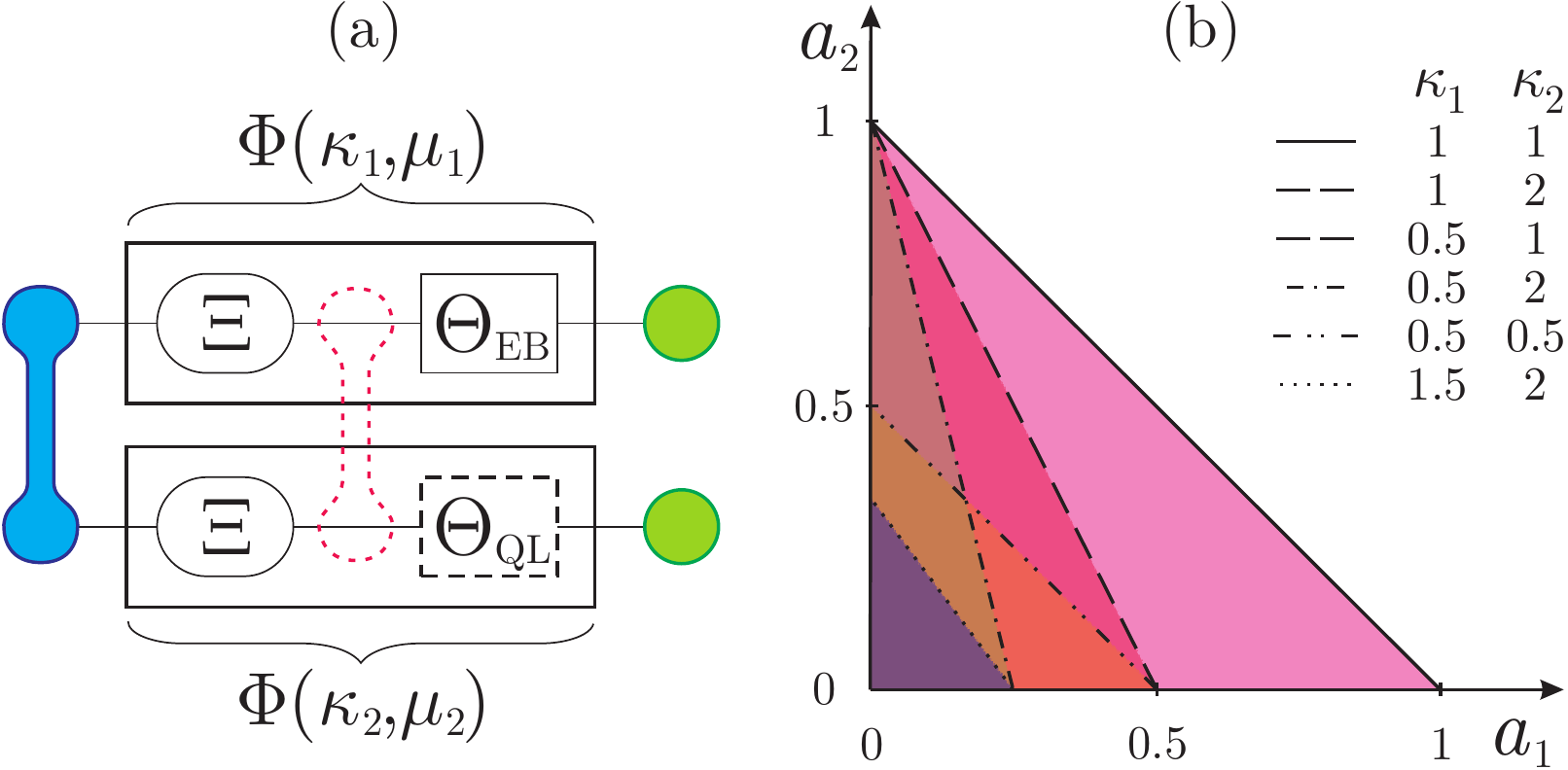}
\caption{\label{figure2} (Color online) (a) Decomposition of the
local channel $\Phi(\kappa_1,\mu_1)\otimes\Phi(\kappa_2,\mu_2)$.
Scaling map $\Xi^{\otimes 2}$ transforms any gaussian input into a
valid quantum state (red dotted line). Channel $\Theta_{\rm EB}$
is entanglement breaking, channel $\Theta_{\rm QL}$ is quantum
limited. (b) Regions of additional noises $a_1$, $a_2$ in the
channel $\Phi(\kappa_1,\mu_1)\otimes\Phi(\kappa_2,\mu_2)$, where
the entanglement of high energy gaussian states is preserved.}
\end{figure}

\begin{corollary}
\label{corollary-3} The channel $\Phi(\kappa_1,\mu_1) \otimes
\Phi(\kappa_2,\mu_2)$ annihilates entanglement of all two-mode
gaussian states if and only if $\kappa_1 \mu_2 + \kappa_2 \mu_1
\ge (\kappa_1+\kappa_2)/2$.
\end{corollary}
\begin{proof}
The sufficiency follows from Corollary~\ref{corollary-2} applied
to the case $N=2$. The necessity follows from Simon's criterion
\cite{simon-2000} applied to the two-mode squeezed vacuum with the
covariance matrix \eqref{covariance-squeezed}, where $r
\rightarrow \infty$.
\end{proof}

Corollary~\ref{corollary-3} is followed by the observation: if
$\Phi(\kappa_1,\mu_1)$ is a quantum limited attenuator, then
$\Phi(\kappa_2,\mu_2)$ must be an entanglement breaking channel to
annihilate entanglement of all gaussian states. This fact is in
agreement with the experimental entanglement detection for a
two-mode squeezed state, one half of which is subjected to a near
quantum-limited amplification (i.e. $\Phi(\kappa_1,\mu_1)={\rm
Id}$ and $\Phi(\kappa_2,\mu_2)$ introduces a low noise)
\cite{pooser}. However, if $\Phi(\kappa_1,\mu_1)$ is a quantum
limited amplifier, that property does not hold anymore.
Fig.~\ref{figure2}(b) illustrates additional noises $a_1$ and
$a_2$ in the channel $\Phi(\kappa_1,\mu_1) \otimes
\Phi(\kappa_2,\mu_2)$ that can be tolerated by gaussian entangled
states.

{\it Non-gaussian input states}. Further, we demonstrate that the
entanglement of some low-energy non-gaussian states can be more
robust to attenuation and amplification than that of gaussian
ones.

\begin{proposition}
\label{proposition-2} The channel $\Phi(\kappa_1,\mu_1) \otimes
\Phi(\kappa_2,\mu_2)$ is not entanglement annihilating under the
following conditions:

\noindent {\rm (i)} $\kappa_1<1$, $\kappa_2<1$,
\[a_1<\frac{\kappa_1(1+a_2)}{2(1+a_2)-\kappa_2}, \qquad
a_2<\frac{\kappa_2(1+a_1)}{2(1+a_1)-\kappa_1};\]

\noindent  {\rm (ii)} $\kappa_1<1$, $\kappa_2\ge 1$,
\[a_1<\frac{\kappa_1(\kappa_2+a_2)}{\kappa_2+2a_2}, \qquad
a_2<1-\kappa_2 \frac{1+a_1-\kappa_1}{2(1+a_1)-\kappa_1}; \]

\noindent  {\rm (iii)} $\kappa_1 \ge 1$, $\kappa_2 \ge 1$,
\[a_1 < 1-\frac{\kappa_1 a_2}{\kappa_2+2a_2}, \qquad a_2 <
1-\frac{\kappa_2 a_1}{\kappa_1+2a_1}. \]

\end{proposition}
\begin{proof}
It suffices to find a two-mode state $\ket{\psi}$ such that
$\left(\Phi(\kappa_1,\mu_1)\otimes\Phi(\kappa_2,\mu_2)\right)[\ket{\psi}\bra{\psi}]$
is entangled for parameters $\kappa_{1,2}$ and $a_{1,2}$
satisfying (i)--(iii). Let $\ket{\psi}=
[2(1-e^{-|\gamma|^2})]^{-1/2}(\ket{\gamma}\ket{0} -
\ket{0}\ket{\gamma})$, its energy
$\cE=(1-e^{-|\gamma|^2})^{-1}|\gamma|^2 \rightarrow 1$ when
$|\gamma| \rightarrow 0$. From a series of powerful entanglement
detection techniques
\cite{shchukin-2005,hillery-2006,sperling-2009,filippov-manko-2009,guhne-toth-2009,gabriel-2011,zhang-2013}
we choose \cite{zhang-2013} and modify it to obtain the following
witness:
\begin{equation}
\label{witness} W_{\lambda} = \int \frac{d^2\alpha}{\pi}
\frac{d^2\beta}{\pi} e^{\lambda(|\alpha|^2+|\beta|^2)}
\ket{\alpha}\bra{\beta} \otimes \ket{\beta}\bra{\alpha}.
\end{equation}

\noindent For all pure factorized states
$\ket{\xi}\otimes\ket{\upsilon}$ we have $\tr{W_{\lambda}
\ket{\xi}\bra{\xi} \otimes \ket{\upsilon}\bra{\upsilon}} = \left|
\int \frac{d^2\alpha}{\pi} e^{\lambda|\alpha|^2} \ip{\xi}{\alpha}
\ip{\alpha}{\upsilon} \right|^2 \ge 0$, whereas
$\tr{W_{\lambda}\varrho} < 0$ indicates entanglement of $\varrho$.
If $\lambda>0$, the operator $W_{\lambda}$ becomes unbounded but
its average with the output state can still be finite and negative
(indication of entanglement). A straightforward integration with
the Kraus operators \eqref{kraus} and the witness operator
\eqref{witness} yields
\begin{align}
& {\rm tr} \left\{ W_{\lambda}
\left(\Phi(\kappa_1,\mu_1)\otimes\Phi(\kappa_2,\mu_2)\right)[\ket{\psi}\bra{\psi}] \right\} \nonumber\\
& = \left[ \left(1-e^{-|\gamma|^2}\right) \left( \tau_1 \tau_2
(1-\lambda)^2 - (\tau_1 - 1)(\tau_2 - 1)
\right) \right]^{-1} \nonumber\\
& \times \bigg\{ \exp \left[ - \frac{\eta_1 \tau_1 \left(
1-\lambda(2-\lambda)\tau_2 \right) }{\tau_1 \tau_2 (1-\lambda)^2 -
(\tau_1 - 1)(\tau_2 - 1)}\, |\gamma|^2 \right] \nonumber\\
& + \exp \left[ - \frac{\eta_2 \tau_2 \left(
1-\lambda(2-\lambda)\tau_1 \right) }{\tau_1 \tau_2 (1-\lambda)^2 -
(\tau_1 - 1)(\tau_2 - 1)}\, |\gamma|^2 \right] \nonumber\\
& - 2 \exp \left[ - \left( 1 - \frac{\sqrt{\eta_1 \tau_1 \eta_2
\tau_2} \, (1-\lambda) }{\tau_1 \tau_2 (1-\lambda)^2 - (\tau_1 -
1)(\tau_2 - 1)} \right) |\gamma|^2 \right] \bigg\}, \nonumber
\end{align}

\noindent which is justified for $\lambda < \lambda_0 = 1 -
\sqrt{(\tau_1 - 1)(\tau_2 - 1)/\tau_1 \tau_2}$. The average value
obtained takes negative values in the widest region of parameters
$\eta_{1,2},\tau_{1,2}$ if $|\gamma| \rightarrow 0$. In this case,
the output state is entangled if there exists a solution
$\lambda_{\ast}$ of the inequality $2[\tau_1 \tau_2 (1-\lambda)^2
- (\tau_1 - 1)(\tau_2 - 1) - \sqrt{\eta_1 \tau_1 \eta_2 \tau_2} \,
(1-\lambda)]<\eta_1 \tau_1 + \eta_2 \tau_2
-\lambda(2-\lambda)\tau_1 \tau_2 (\eta_1 + \eta_2)$ such that
$\lambda_{\ast} < \lambda_0$. The reader will have no difficulty
in showing that such a solution $\lambda_{\ast}$ exists if
$2-\eta_1-\tau_2(2-\eta_1-\eta_2) > 0$ and
$2-\eta_2-\tau_1(2-\eta_1-\eta_2) > 0$. Substituting expressions
\eqref{k1-k2} for $\eta_{1,2}$ and $\tau_{1,2}$ yields formulas
(i)--(iii).
\end{proof}

\begin{figure}
\includegraphics[width=8.5cm]{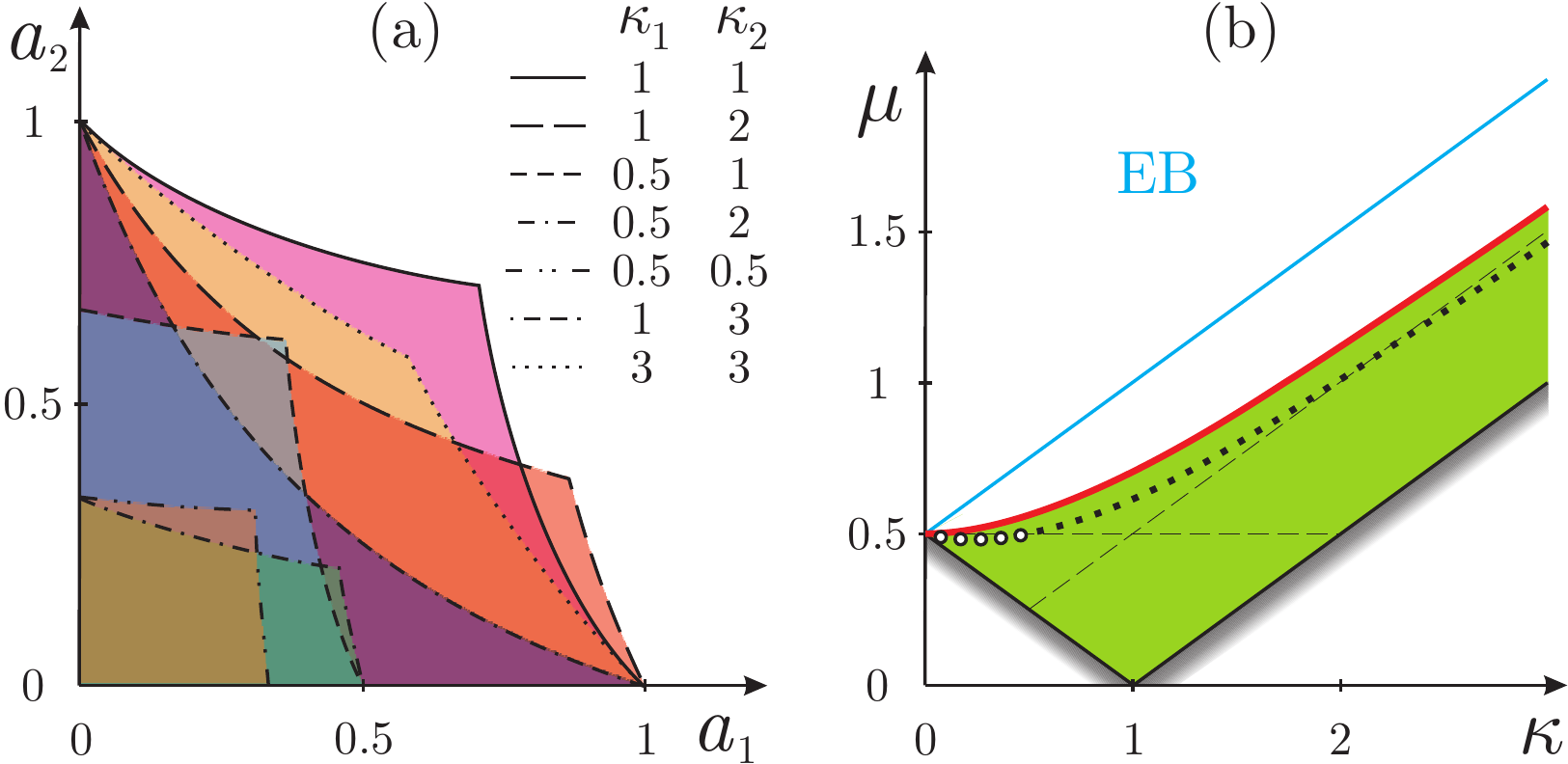}
\caption{\label{figure3} (Color online) (a) Regions of additional
noises $a_1$, $a_2$ in the channel
$\Phi(\kappa_1,\mu_1)\otimes\Phi(\kappa_2,\mu_2)$, where the
entanglement of low-energy state $\ket{\psi} \propto
\ket{\gamma}\ket{0} - \ket{0}\ket{\gamma}$ is preserved. (b)
Channel $\Phi(\kappa,\mu)^{\otimes 2}$ does not annihilate
entanglement of the state $\ket{\psi} \propto \ket{\gamma}\ket{0}
- \ket{0}\ket{\gamma}$ for points $(\kappa,\mu)$ below the red
solid line (dashed lines are asymptotes). Dotted line is the
result of Ref. \cite{sabapathy-2011} obtained for the state
$\frac{1}{\sqrt{2}}(\ket{n}\ket{0} + \ket{0}\ket{n})$, $n=5$, and
circles represent the points where gaussian states outperform that
result.}
\end{figure}

The result of Proposition~\ref{proposition-2} is depicted in
Fig.~\ref{figure3}(a). If $\Phi(\kappa_1,\mu_1)$ is a quantum
limited attenuator, then the entanglement of the non-gaussian
state $\ket{\psi} \propto \ket{\gamma}\ket{0} -
\ket{0}\ket{\gamma}$ is preserved in a narrower range of
parameters $\kappa_2,a_2$ than for the gaussian state with large
squeezing. Thus, gaussian state entanglement is favorable for
transmission through lossy channel with high asymmetry in noises.
On the contrary, if the losses are quite similar, then the
non-gaussian state $\ket{\psi}$ is at an advantage. As far as
amplifiers are concerned, non-gaussian states undoubtedly
outperform gaussian ones. Moreover, if $\Phi(\kappa_1,\mu_1)$ and
$\Phi(\kappa_2,\mu_2)$ are amplifiers with $|\kappa_1-\kappa_2|
\le 2$ and one of them is quantum limited, then the other has to
be entanglement breaking to destroy entanglement of the state
$\ket{\psi}$.

\begin{corollary}
\label{corollary-4} The channel $\Phi(\kappa,\mu) \in \cC$ is not
$N\text{-LEA}$ for any $N=2,3,\ldots$ if the total noise level
satisfies $\mu < \frac{1}{2} \sqrt{\kappa^2+1}$.
\end{corollary}
\begin{proof}
Applying the result of Proposition~\ref{proposition-2} for
$\kappa_1=\kappa_2=\kappa$ and $a_1=a_2=a$ and solving the
corresponding inequality with respect to $a$ gives $a <
\frac{1}{2}(\sqrt{\kappa^2+1} - |\kappa-1|)$. The relation $\mu =
|\kappa-1|/2 + a$ leads to $\mu < \frac{1}{2} \sqrt{\kappa^2+1}$,
which indicates when $\Phi(\kappa,\mu)$ is not 2-LEA and,
consequently, not $N$-LEA.
\end{proof}

If $\mu < \frac{1}{2} \sqrt{\kappa^2+1}$, the state $\ket{\psi}
\propto \ket{\gamma}\ket{0} - \ket{0}\ket{\gamma}$ remains
entangled irrespective to the value of $\gamma \ne 0$, i.e. for
all energies $\cE\in(1,+\infty)$. This fact follows from the
expression ${\rm tr} \left\{ W_{\lambda} \Phi(\kappa,\mu)^{\otimes
2}[\ket{\psi}\bra{\psi}] \right\}$ containing the difference of
two exponents only. Note that Corollary~\ref{corollary-4} could be
proven by considering input states $\ket{\psi_n} =
\frac{1}{\sqrt{2}}(\ket{n}\ket{0}-\ket{0}\ket{n})$ and the witness
$\widetilde{W}_{\lambda} = \sum_{i,j=0}^{\infty} \lambda^{i+j}
\ket{i}\bra{j} \otimes \ket{j}\bra{i}$ when $\lambda$ tends to
$\tau/(\tau - 1)$. It is shown in Ref.~\cite{sabapathy-2011} that
the state $\frac{1}{\sqrt{2}}(\ket{n}\ket{0} + \ket{0}\ket{n})$
outperforms robustness of gaussian states with respect to
homogeneous attenuation and amplification
$\Phi(\kappa,\mu)\otimes\Phi(\kappa,\mu)$ for some region of
parameters $\kappa$ and $\mu$, however, there was no evidence that
these states outperform gaussian states for $\kappa<0.43$ even if
$n \rightarrow \infty$ [see Fig.~\ref{figure3}(b)]. We have just
used a stricter entanglement detection method and proved the
existence of non-gaussian states with little energy that
outperform gaussian states for all values of $\kappa > 0$.

To conclude, we have analyzed the noises that accompany
attenuation and amplification and impose fundamental limitations
on the performance of entanglement-assisted devices. The important
practical conclusion is that gaussian states of high energy are
quite robust to lossy channels with high asymmetry in the noises,
whereas non-gaussian states are more robust in the case of similar
attenuations. Gaussian state entanglement cannot withstand
amplification with power gain $2$ ($\approx 3$~dB), whereas
non-gaussian states of small energy can preserve the entanglement
for arbitrarily large power gains if the introduced noise is
sufficiently small.

{\it Acknowledgments}. The authors thank the referees for their
useful comments. S.N.F. is grateful to A.S. Holevo and
participants of seminar ``Quantum probability, statistics,
information'' at Steklov Mathematical Institute of the Russian
Academy of Sciences for fruitful discussions. S.N.F. acknowledges
partial support from the Russian Foundation for Basic Research
under Project No. 14-07-00937-a. M.Z. acknowledges support from
grants APVV-0808-12 (QIMABOS), RAQUEL (7FP STREP), and GACR
Project P202/12/1142.


\begin{thebibliography}{99}
\bibitem{horodecki-2009}
R. Horodecki, P. Horodecki, M. Horodecki, and K. Horodecki, Rev.
Mod. Phys. {\bf 81}, 865 (2009).

\bibitem{giovannetti-lloyd-maccone-2011}
V. Giovannetti, S. Lloyd, and L. Maccone, Nat. Photonics {\bf 5},
222 (2011).

\bibitem{banaszek-2009}
K. Banaszek, R. Demkowicz-Dobrza\'{n}ski, and I. A. Walmsley,
Nature Photonics {\bf 3}, 673 (2009).

\bibitem{serafini-2004}
A. Serafini, F. Illuminati, M. G. A. Paris, and S. De Siena, Phys.
Rev. A {\bf 69}, 022318 (2004).

\bibitem{manko-2005}
V. I. Man'ko and O. V. Pilyavets, J. Russ. Laser Res. {\bf 26},
259 (2005).

\bibitem{adesso-illuminati-2007}
G. Adesso and F. Illuminati, J. Phys. A: Math. Theor. {\bf 40}
7821 (2007).

\bibitem{an-2007}
J.-H. An and W.-M. Zhang, Phys. Rev. A {\bf 76}, 042127 (2007).

\bibitem{paz-2008}
J. P. Paz and A. J. Roncaglia, Phys. Rev. Lett. {\bf 100}, 220401
(2008).

\bibitem{isar-2009}
A. Isar, Phys. Scr. {\bf T135}, 014033 (2009).

\bibitem{vasile-2009}
R. Vasile, S. Olivares, M. G. A. Paris, and S. Maniscalco, Phys.
Rev. A {\bf 80}, 062324 (2009).

\bibitem{galve-2010}
F. Galve, G. L. Giorgi, and R. Zambrini, Phys. Rev. A {\bf 81},
062117 (2010).

\bibitem{argentieri-2011}
G. Argentieri, F. Benatti, R. Floreanini, and U. Marzolino, Int.
J. Quantum Inf. {\bf 9}, 1745 (2011).

\bibitem{barbosa-2011}
F. A. S. Barbosa, A. J. de Faria, A. S. Coelho, K. N. Cassemiro,
A. S. Villar, P. Nussenzveig, and M. Martinelli, Phys. Rev. A {\bf
84}, 052330 (2011).

\bibitem{buono-2012}
D. Buono, G. Nocerino, A. Porzio, and S. Solimeno, Phys. Rev. A
{\bf 86}, 042308 (2012).

\bibitem{marzolino-2013}
U. Marzolino, Europhys. Lett. {\bf 104}, 40004 (2013).

\bibitem{rebon-2014}
L. Reb\'{o}n, N. Canosa, and R. Rossignoli, Phys. Rev. A {\bf 89},
042312 (2014).

\bibitem{aolita-2014}
L. Aolita, F. de Melo, L. Davidovich, arXiv:1402.3713 [quant-ph].

\bibitem{caves-1982}
C. M. Caves, Phys. Rev. D {\bf 26}, 1817 (1982).

\bibitem{scarani-2005}
V. Scarani, S. Iblisdir, N. Gisin, and A. Ac\'{i}n, Rev. Mod.
Phys. {\bf 77}, 1225 (2005).

\bibitem{weedbrook-2012}
C. Weedbrook, S. Pirandola, R. Garc\'{i}a-Patr\'{o}n, N. J. Cerf,
T. C. Ralph, J. H. Shapiro, and S. Lloyd, Rev. Mod. Phys. {\bf
84}, 621 (2012).

\bibitem{holevo-2008}
A. S. Holevo, Probl. Inf. Transm. {\bf 44}, 3 (2008).

\bibitem{holevo-1998}
A. S. Holevo, Russ. Math. Surveys {\bf 53}, 1295 (1998).

\bibitem{horodecki-2003}
M. Horodecki, P. W. Shor, and M. B. Ruskai, Rev. Math. Phys. {\bf
15}, 629 (2003).

\bibitem{sabapathy-2011}
K. K. Sabapathy, J. S. Ivan, and R. Simon, Phys. Rev. Lett. {\bf
107}, 130501 (2011).

\bibitem{allegra-2010}
M. Allegra, P. Giorda, and M. G. A. Paris, Phys. Rev. Lett. {\bf
105}, 100503 (2010).

\bibitem{adesso-2011}
G. Adesso, Phys. Rev. A {\bf 83}, 024301 (2011).

\bibitem{lee-2011}
J. Lee, M. S. Kim, and H. Nha, Phys. Rev. Lett. {\bf 107}, 238901
(2011).

\bibitem{allegra-2011}
M. Allegra, P. Giorda, and M. G. A. Paris, Phys. Rev. Lett. {\bf
107}, 238902 (2011).

\bibitem{holevo-werner-2001}
A. S. Holevo and R. F. Werner, Phys. Rev. A {\bf 63}, 032312
(2001).

\bibitem{caruso-2006}
F. Caruso, V. Giovannetti, and A. S. Holevo, New J. Phys. {\bf 8},
310 (2006).

\bibitem{holevo-2007}
A. S. Holevo, Probl. Inf. Transm. {\bf 43}, 1 (2007).

\bibitem{ivan-2011}
J. S. Ivan, K. K. Sabapathy, and R. Simon, Phys. Rev. A {\bf 84},
042311 (2011).

\bibitem{moravcikova-ziman-2010}
L. Morav\v{c}\'{i}kov\'{a} and M. Ziman, J. Phys. A: Math. Theor.
{\bf 43}, 275306 (2010).

\bibitem{filippov-rybar-ziman-2012}
S. N. Filippov, T. Ryb\'{a}r, and M. Ziman, Phys. Rev. A {\bf 85},
012303 (2012).

\bibitem{filippov-ziman-2013}
S. N. Filippov and M. Ziman, Phys. Rev. A {\bf 88}, 032316 (2013).

\bibitem{filippov-melnikov-ziman-2013}
S. N. Filippov, A. A. Melnikov, and M. Ziman, Phys. Rev. A {\bf
88}, 062328 (2013).

\bibitem{holevo-giovannetti}
A. S. Holevo and V. Giovannetti,  Rep. Prog. Phys. {\bf 75},
046001 (2012).

\bibitem{simon-2000}
R. Simon, Phys. Rev. Lett. {\bf 84}, 2726 (2000).

\bibitem{pooser}
R. C. Pooser, A. M. Marino, V. Boyer, K. M. Jones, and P. D. Lett,
Phys. Rev. Lett. {\bf 103}, 010501 (2009).

\bibitem{shchukin-2005}
E. Shchukin and W. Vogel, Phys. Rev. Lett. {\bf 95}, 230502
(2005).

\bibitem{hillery-2006}
M. Hillery and M. S. Zubairy, Phys. Rev. Lett. {\bf 96}, 050503
(2006).

\bibitem{sperling-2009}
J. Sperling and W. Vogel, Phys. Rev. A {\bf 79}, 022318 (2009).

\bibitem{filippov-manko-2009}
S. N. Filippov and V. I. Man'ko, J. Russ. Laser Res. {\bf 30}, 55
(2009).

\bibitem{guhne-toth-2009}
O. G\"{u}hne, G. T\'{o}th, Phys. Rep. {\bf 474}, 1 (2009).

\bibitem{gabriel-2011}
A. Gabriel, M. Huber, S. Radic, and B. C. Hiesmayr, Phys. Rev. A
{\bf 83}, 052318 (2011).

\bibitem{zhang-2013}
C. Zhang, S. Yu, Q. Chen, and C. H. Oh, Phys. Rev. Lett. {\bf
111}, 190501 (2013).

\end{thebibliography}
\end{document}